\documentclass[conference]{IEEEtran}
\IEEEoverridecommandlockouts
\usepackage{cite}
\usepackage{amsmath,amssymb,amsfonts}
\usepackage{algorithm}
\usepackage{algpseudocode}
\usepackage{graphicx}
\usepackage{textcomp}
\usepackage{xcolor}
\def\BibTeX{{\rm B\kern-.05em{\sc i\kern-.025em b}\kern-.08em
    T\kern-.1667em\lower.7ex\hbox{E}\kern-.125emX}}

\usepackage{topcapt,booktabs}
\usepackage{tikz}
\usepackage{pgfplots} 
\usepackage{listings}
\usepackage{tabularx}
\usepackage{colortbl}
\usepackage{multirow}
\usepackage{amsthm}
\pgfplotsset{compat=newest}
\usetikzlibrary{shapes,arrows,positioning}
\usetikzlibrary{shapes.arrows,patterns}
\usepackage{arydshln}

\usepackage{pgfplots}
\pgfplotsset{width=10cm,compat=1.9}

\usepackage{soul}

\newtheorem{theorem}{Theorem}[section]
\newtheorem{claim}[theorem]{Claim}

\begin{document}

\title{Automatic Generation of Complete Polynomial Interpolation Hardware Design Space}

\author{\IEEEauthorblockN{Bryce Orloski}
\IEEEauthorblockA{Numerical Hardware Group \\
Intel Corporation\\
bryce.orloski@intel.com}
\and
\IEEEauthorblockN{Samuel Coward}
\IEEEauthorblockA{Numerical Hardware Group \\
Intel Corporation \\
samuel.coward@intel.com}
\and
\IEEEauthorblockN{Theo Drane}
\IEEEauthorblockA{Numerical Hardware Group \\
Intel Corporation\\
theo.drane@intel.com}
}

\maketitle

\begin{abstract}
Hardware implementations of complex functions regularly deploy piecewise polynomial approximations. This work determines the complete design space of piecewise polynomial approximations meeting a given accuracy specification. Knowledge of this design space determines the minimum number of regions required to approximate the function accurately enough and facilitates the generation of optimized hardware which is competitive against the state of the art. Targeting alternative hardware technologies simply requires a modified decision procedure to explore the space.
\end{abstract}


\section{Introduction}
A common challenge in hardware design is how best to compute complex functions such as the reciprocal, sine, cosine etc. Hardware algorithms to compute elementary functions generally use one of the following techniques: digit-recurrence~~\cite{Muller1994Bkm:AFunctions}, CORDIC (COordinate Rotation Digital Computer) \cite{Volder1959TheTechnique} or piecewise polynomials \cite{Tang1991Table-lookupAnalysis}. In this paper we focus on piecewise quadratic or linear implementations. The degree and number of polynomials used to approximate a function depends on the required precision and target error bound. 
For binary32 interpolation, quadratic is usually sufficient.

We study the following question. Given a fixed-point function to approximate to within some error bound, what is the complete design space of all feasible piecewise quadratic/linear approximations, constrained only by the underlying architecture, described in Figure \ref{fig:interp_arch}. Knowledge of the complete design space allows us to tailor the design space exploration for different hardware targets, without needing to re-generate the design space. Interpreting mathematical bounds as polynomial design space constraints is not a novel concept \cite{Drane2012CorrectlyMultiply-add} but is extended and generalised in this work.

Automatic tools for generating efficient piecewise polynomial hardware already exist \cite{Lee2005OptimizingEvaluation, Detrey2005Table-basedEvaluation,Strollo2011ElementaryApproximations}, one example being FloPoCo \cite{deDinechin2011CustomGenerator,DeDinechin2010AutomaticEvaluation} which uses Sollya \cite{Chevillard2010Sollya:Codes} to generate its polynomial approximations. Sollya uses a modified Remez algorithm~\cite{Brisebarre2007EfficientL-approximations}, which computes minmax polynomial approximations subject to the constraints of finite precision coefficients. Such an approach explores a constrained design space, allowing it to quickly generate high precision approximations. 

We present the mathematics to generate the complete design space of piecewise polynomial approximations to a given function. An example decision procedure is presented along with the resulting hardware. 

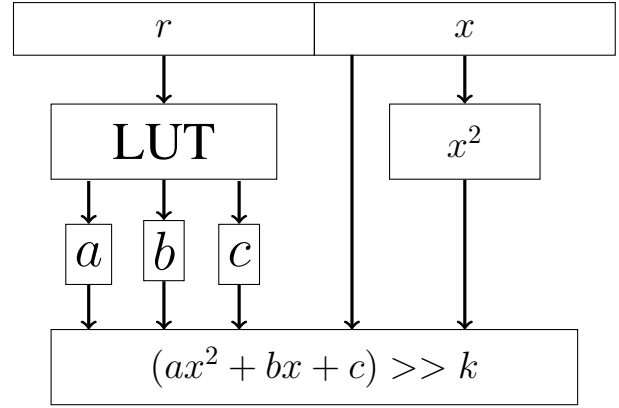
\begin{figure}
    \centering
    \begin{tikzpicture}

\node [shape=rectangle,draw = black, minimum width = 4cm,minimum height=0.7cm] at (0,6.5) (r) {\Large $r$};
\node [shape=rectangle,draw = black, minimum width = 4cm,minimum height=0.7cm] at (4,6.5) (x) {\Large $x$};

\node [shape=rectangle,draw = black, minimum width = 3cm,minimum height=1cm] at (0,5) (lut) {\huge LUT};
\node [shape=rectangle,draw=black, minimum height = 0.8cm] at (-1,3.5) (a) {\huge $a$};
\node [shape=rectangle,draw=black, minimum height = 0.8cm] at (0,3.55) (b) {\huge $b$};
\node [shape=rectangle,draw=black, minimum height = 0.8cm] at (1,3.5)  (c) {\huge $c$};

\node [shape=rectangle,draw = black,minimum height=1cm, minimum width = 2cm] at (4,5) (square) {\Large$x^2$};


\node [shape=rectangle,draw = black,minimum height=1cm, minimum width = 7cm] at (2,2) (poly) {\Large{$(ax^2+bx+c) >>k$}};

\node[] (left_x) at (2.5,6.28) {};
\draw [->,very thick] (left_x) edge (2.5,2.5);
\draw [->,very thick] (x) edge (square);
\draw [->,very thick] (r) edge (lut);

\node[] (left_lut) at (-1,4.6) {};
\node[] (right_lut) at (1,4.6) {};
\draw [->,very thick] (left_lut) edge (a);
\draw [->,very thick] (lut) edge (b);
\draw [->,very thick] (right_lut) edge (c);

\draw [->,very thick] (a) edge (-1,2.5);
\draw [->,very thick] (b) edge (0,2.5);
\draw [->,very thick] (c) edge (1,2.5);

\draw [->,very thick] (square) edge (4,2.5);
\end{tikzpicture}
    \caption{Quadratic interpolation hardware architecture. The most significant bits, $r$, of the input are passed to a lookup table (LUT). The least significant bits, $x$, are used in the polynomial evaluation \cite{Tang1991Table-lookupAnalysis}.}
    \label{fig:interp_arch}
\end{figure}

\section{Design Space Generation} \label{sec:dsg}
The target function, $f$, and accuracy are specified via input and output precisions along with upper and lower bounds across all inputs. Using upper and lower bounds provides maximal flexibility and can even accommodate asymmetric error bounds on a function. Determining the space of feasible quadratic interpolations is reduced to a series of inequalities. 

We use fixed-point notation $n.m$ to denote a format with $n$ integral bits and $m$ fractional bits.  Given $f : n.m \rightarrow p.q$ and integer upper and lower bound functions $u,l: n+m \rightarrow p+q$. For a fixed number of lookup bits, $R$, define the following.
\begin{align*}
    Z &= z_{n-1}...z_0.z_{-1}...z_{-m} &(\textrm{input fixed-point value})\\
    r &= z_{n-1}...z_{n-R} &(\textrm{unsigned integer})\\
    x &= z_{n-R-1}...z_{-m} &(\textrm{unsigned integer}) 
\end{align*}
Let $l_R(r,x) = l(\{r,x\})$ and $u_R(r,x) = u(\{r,x\})$, where $\{,\}$ denotes concatenation.
Under these definitions the bounds satisfy $2^{-q} l_R(r,x) \leq f(Z) \leq 2^{-q} u_R(r,x)$. 
Fixing an $R$ also fixes an interval for $x\in I=[0,2^{n+m-R}-1]$, so for a given value of $r < 2^R$, a feasible quadratic is defined by the quad $(a,b,c,k)$, where $k$ is the difference between the polynomial evaluation precision and output precision, which satisfies:
\begin{equation*} \label{eqn:initial_constraint}
     \forall x \in I, \, l_R(r,x) \leq \left \lfloor \frac{ax^2 + bx + c}{2^k}  \right \rfloor \leq u_R(r,x)
\end{equation*}
or equivalently (the $x\in I$ being implied in the remainder)
\begin{equation*} 
    \forall x \hspace{1em} l_R(r,x) \leq \frac{ax^2 + bx + c}{2^k} < u_R(r,x) + 1.
\end{equation*}
Rearranging gives necessary and sufficient existence conditions for $c$.
\begin{equation}\label{eqn:c_bounds}
    2^k l_R(r,x) - ax^2 - bx \leq c < 2^k(u_R(r,x) + 1) - ax^2 - bx 
\end{equation}
If a feasible $c$ exists then $\forall x,y$:
\begin{equation} \label{eqn:c_existence}
        2^k l_R(r,x) - ax^2 - bx < 2^k(u_R(r,y) + 1) - ay^2 - by. 
\end{equation}
Introducing, 
\[d(r,x,y)= \frac{u_R(r,y) + 1 - l_R(r,x)}{y-x},\]
reminiscent of a numerical derivative, we can bound $b$, under some assumptions. Eqn. \ref{eqn:c_existence} is trivially true for $x=y$.
\begin{align}
    x < y \Rightarrow b < 2^k d(r,x,y) - a(x+y) \label{eqn:b_upper}\\
    x > y \Rightarrow b > 2^k d(r,x,y) - a(x+y) \label{eqn:b_lower}
\end{align}
If a feasible $b$ exists then $\forall x<y$ and $w<z$:
\begin{equation} \label{eqn:b_condtion}
    2^k d(r,y,x) - a(x+y) < 2^k d(r,w,z) - a(z+w).
\end{equation}
Eqns. \ref{eqn:c_existence} \& \ref{eqn:b_condtion} are necessary and sufficient conditions on existence since $k$ can be increased until the intervals contain an integer. Finally we bound $a$, using expressions that are closely related to numerical second derivatives. 
\begin{align}
    x+y = w+z &: d(r,y,x) < d(r,w,z) \\
    x+y < w+z &: \frac{a}{2^k} < \frac{d(r,w,z)-d(r,y,x)}{w+z-x-y} \label{eqn:a_upper}\\
    x+y > w+z &: \frac{a}{2^k} > \frac{d(r,y,x)-d(r,w,z)}{x+y-w-z} \label{eqn:a_lower}
\end{align}
Introducing further definitions,
\begin{equation*}
    M(r,t) = \max_{\substack{x<y\\ x+y=t}} d(r,y,x),\quad 
    m(r,t) = \min_{\substack{w<z\\ w+z=t}} d(r,w,z).
\end{equation*}
The necessary and sufficient conditions for the existence of a feasible polynomial over a given region specified by $r$ are:
\begin{align}
    \forall t \, M(r,t) &< m(r,t) \textrm{ and} \label{eqn: necc_and_suff_1}\\
    \max_{t<s} \frac{M(r,s) - m(r,t)}{s-t} &< \min_{t<s} \frac{m(r,s) - M(r,t)}{s-t}. \label{eqn: necc_and_suff_2}
\end{align}
These bounds are intuitive because we bound the $b$ coefficient by something resembling a first derivative and $a$ by a second derivative term.

To generate the design space for a value of $R$, we test whether Eqns. \ref{eqn: necc_and_suff_1} \& \ref{eqn: necc_and_suff_2} hold for all $r\in [0, 2^R - 1]$. Satisfiability implies existence of at least one feasible quadratic in each region, and we proceed to establish a dictionary of coefficients. Eqns. \ref{eqn:a_upper} \& \ref{eqn:a_lower} determine an interval $[a_0,a_1]$ of valid $a$ values for $k=0$. Then $\forall a \in [a_0, a_1]$, we solve Eqns. \ref{eqn:b_upper} \& \ref{eqn:b_lower} to generate the interval $[b_0,b_1]$ of valid $b$ values, increasing $k$ if necessary to ensure we obtain at least one valid $b$ in each region. Lastly, for each valid $(a,b)$ pair we solve Eqn. \ref{eqn:c_bounds}, yielding an interval of valid $c$ values. Across all regions $k$ is constant. The result of this process is a nested dictionary of valid polynomial coefficients for fixed $k$ and $R$ values.

This dictionary represents the complete design space of feasible quadratic polynomials that satisfy the given upper and lower bound functions across the complete input space. If $\forall r\in[0,2^R - 1]$, $0\in[a_0,a_1]$ then a piecewise linear approximation will suffice, resulting in smaller and faster hardware.

\subsection{Performance}
Design space generation involves many 2-D searches across a potentially large search space, evaluating expressions of the form $\max_{x<y} D(x,y)$, where $D(x,y)~=~\frac{g(y)-h(x)}{y-x}$, for some $g$ and $h$, or the minimum of such expressions. To improve scalability we optimise these searches. A naive implementation would iterate across the complete 2-D space. In practice, we skip iterations of this search due to claim \ref{claim:skip_iter}.

\begin{claim}\label{claim:skip_iter}
Let $(x',y')$ be the arguments which maximise $D(x,y)$ across all $x<x'$. Then for $x>x'$,
\[\hspace{0.5em} D(x',y') \leq \frac{h(x) - h(x')}{x-x'} \Rightarrow \nexists y\textrm{ s.t. }D(x,y) > D(x',y').\]

\end{claim}

\begin{proof} 
Suppose $\exists x>x' \textrm{ and } y>x \textrm{ s.t. } D(x,y)>D(x',y')$, it follows that $D(x,y)>D(x',y)$, since $D(x',y')$ maximal. Expanding the definition of $D$ and re-arranging,
\begin{equation*}
    (x-x')g(y) + x' h(x) > y(h(x) - h(x')) + x h(x').
\end{equation*}
Subtracting $x'h(x')$ from both sides and re-arranging,
\begin{equation*} 
    D(x',y) > \frac{h(x) - h(x')}{x-x'}.
\end{equation*}
Since $D(x',y') \geq D(x',y)$, we have exactly the converse of the condition in our claim.  
\end{proof}

Using this optimization the runtime of the design space generation is five times faster for a 16 bit reciprocal approximation running single threaded on an Intel Xeon E3-1270 CPU. 

There is a computational tradeoff between the number of regions the input interval is sub-divided into ($R$), versus the input range that each polynomial must span, corresponding to the number of inputs to check for each polynomial. Empirical results for a 16 bit design suggest the runtime is $\mathcal{O}(R^{-3})$.
The design space generation algorithm scales exponentially in the number of bits of precision so these speedup techniques can improve runtimes in practical cases but do not substantially expand the space of computationally feasible designs.

\section{Design Space Exploration} \label{sec:dse}
Having generated the design space, we must now derive efficient methods to explore this space. The exploration procedure can be tailored to the target hardware technology, one of the major advantages of generating the complete design space. In the procedures presented here we will have a target number of lookup bits to be used, as the optimal lookup table (LUT) size is non-obvious as we shall see in \S \ref{sec:results}.

There are two distinct paths through the hardware design presented in Figure \ref{fig:interp_arch}, one through the square operation and one through the LUT, since these two execute in parallel. In this work we will assume that the square path is critical and target optimized ASIC designs. Optimisations performed on one part of the design restrict the available optimisations in other parts of the design, so decision procedure tuning is important. The decision procedure used in this work is the following.

\begin{enumerate}
    \item Minimize $k$ - minimize polynomial evaluation precision
    \item Maximize square input truncation
    \item Maximize linear input truncation
    \item Minimize $a$, then $b$, then $c$ bitwidths
\end{enumerate} 


The procedure begins by minimizing $k$, which is found via Eqns. \ref{eqn:c_existence} \& \ref{eqn:b_condtion}. We then maximise the square truncation, asking what is the maximum integer $i$, such that a valid $a(x[m-1:i])^2 + bx +c$ exists in all regions, where $m$ is the bitwidth of $x$. Intuitively, we often think of higher order terms as correction terms, so we can tolerate some error in them, to gain performance. Such truncation is found in other approaches \cite{Detrey2005Table-basedEvaluation}. Only a subset of the polynomials $(a,b,c)$ for each region can tolerate the error induced by the maximal square truncation, so we discard those that cannot. We then similarly calculate maximal $j$, such that $a(x[m-1:i])^2 + bx[m-1:j] +c$ is still valid, introducing further error and hence further candidates are discarded. Lastly, we minimise the precision required to represent the coefficients, $a$ then $b$ and finally $c$. For each coefficient, we have a set of valid integer values per region, which we separate into positive and negative sets (and take absolute values), then run Algorithm \ref{alg:prec_min} on each set and take the minimum of the two returned precisions. With the precision of the coefficient defined we then prune the dictionary, removing any candidates that require a higher precision. From the remaining feasible polynomials we pick the first polynomial for each region.

\begin{table*}[t]
    \centering
    \caption{Logic synthesis results for minimum obtainable delay target comparing against equivalent Designware components. We select the number of lookup bits (LUB) for the proposed RTL based on the best area-delay product.}
    \begin{tabular}
    {|c|c|c|l|c|c|c|c|c|c|}
        \hline
        \multirow{2}{7em}{Function} & $x\rightarrow y$  & \multicolumn{5}{c|}{Proposed} & \multicolumn{3}{c|}{DesignWare} \\
                 & Num Bits    &         Runtime & \multicolumn{1}{c|}{LUB} & Delay (ns) & Area ($\mu m^2$) & Area $\times$ Delay  & Delay (ns) & Area ($\mu m^2$) & Area $\times$ Delay\\
        \hline
        \multirow{3}{7em}{$0.1y =$\large $\frac{1}{1.x}$}  &   10 $\rightarrow$ 10 & 0.5 sec & 6 (lin) & 0.125 & 43 & \textbf{5.4}   & 0.143  & 79 & 11.3\\
                                                 &  16 $\rightarrow$ 16 & 8.6 sec & 8 (lin) & 0.197 & 290 & 57.2 & 0.189 & 204 & \textbf{38.6} \\
                                                 &  23 $\rightarrow$ 23 & 39 hrs  & 7 (quad) & 0.278 & 689 & 191.5 & 0.275 & 641 & \textbf{176} \\
        \hdashline
        \multirow{3}{7em}{$0.y =\log_2(1.x)$}   &  10 $\rightarrow$ 11 & 0.4 sec & 6 (lin) & 0.112 & 55 & \textbf{6.2} & 0.134 & 80 & 10.8 \\
                                                & 16 $\rightarrow$ 17 & 22 sec  & 8 (lin) & 0.194 & 301 & \textbf{58.4} & 0.226 & 340 & 75.7\\
                                                & 23 $\rightarrow$ 24 & 78 hrs & 7 (quad) & 0.274 & 703 & 193 & 0.281 & 582 & \textbf{164}\\
        \hdashline
        \multirow{2}{7em}{$1.y= 2^{0.x}$} &  10 $\rightarrow$ 10 & 0.4 sec & 5 (lin) & 0.115 & 48  & \textbf{5.5}  & 0.117 & 52  & 6.2 \\
                                            &  16 $\rightarrow$ 16 & 24 sec  & 7 (lin) & 0.190 & 201 & 38.1 & 0.169 & 182 & \textbf{30.7} \\
        \hline        
    \end{tabular}

    \label{tab:results_table}
\end{table*}

\begin{algorithm}
\caption{Precision Minimization Algorithm}\label{alg:prec_min}
\begin{algorithmic}
\Require $S = \{S_r\subseteq \mathbb{N}\,|\, r = 0...2^R-1\}$\\
// \textit{Number of trailing zeros for each element}
\For{$r=0...2^R-1, s\in S_r$}
        \State $T_{r,s} = \max_i \left ( (s>>i)<<i == s \right )$ \Comment{trailing zeros} 
\EndFor
\State $T = \min_{r<2^R} \max_{s\in S_r} T_{r,s}$ \Comment{max valid truncation}\\
// \textit{Calculate the optimal number of zeros to truncate}
\For{$t=0...T$, $r=0...2^R-1$}
        \State $S_{t,r} = \{s \,|\, s\in S_r\textrm{ and }T_{r,s} \geq t\}$ \Comment{prune each set}
        \State $P_{t,r} = \min_{s\in S_{t,r}} (\left \lceil \log_2(s+1) \right \rceil - t)$ \Comment{num bits for $s$}
\EndFor
\State $P = \min_{t\leq T} \max_{r<2^R} P_{t,r}$ \Comment{min precision}
\end{algorithmic}
\end{algorithm}


Alternative decision procedures were explored, such as prioritizing LUT optimsization, but this yielded inferior area-delay profiles for the generated hardware. Further optimizations such as sum of product truncation were also explored but again yielded worse hardware.

\begin{figure}
    \centering
    \begin{tikzpicture}[scale=0.75]
	\begin{axis}[
		xlabel=Delay (ns),
		ylabel=Area $(\mu m^2)$]
	\addplot[color=blue,mark=o] coordinates {
        (0.278,641.550964)
        (0.3 ,474.504482)
        (0.32,359.031601)
        (0.36,370.905840)
        (0.42,367.691040)
        (0.5,339.948000)
        (0.6,293.846400)
        (0.8,226.923839)
	};
	\addplot[color=red,mark=x] coordinates {
        (0.278,689.184724)
        (0.3,488.334962)
        (0.32,372.957841)
        (0.36,319.482721)
        (0.42,370.632240)
        (0.5,333.149040)
        (0.6,258.374160)
        (0.8,233.61335)
	};
	\legend{DesignWare, Proposed}
	\end{axis}
\end{tikzpicture}
    \caption{23 bit reciprocal using 7 lookup bits to compute $\frac{1}{1.x}$, demonstrating competitiveness across the delay spectrum against industrial state of the art.}
    \label{fig:single_prec_recip}
\end{figure}
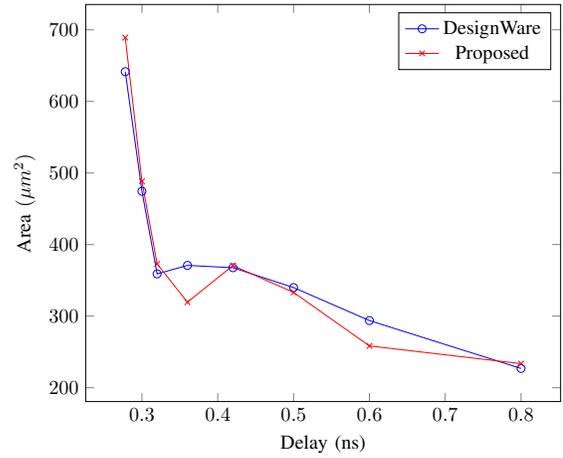
\section{Results} \label{sec:results}
We implemented a piecewise polynomial generating tool using PyPy 3.7. Upper and lower bounding functions are produced using Python's math library or standard integer computations. We automatically generated Register Transfer Level (RTL) implementations of three complex functions: reciprocal, base two logarithm and base two exponential, at relevant precisions and all possible LUT height targets with an accuracy of one unit in the last place (ULP), which matches the default accuracy of FloPoCo \cite{deDinechin2011CustomGenerator} and DesignWare \cite{Synopsys2021DesignS-2021.06-SP2}. Hardware was generated on an Intel Xeon E3-1270 CPU and synthesised using Synopsys Design Compiler for a TSMC 7nm cell library. 

For the reciprocal function, behavioural RTL producing both Round to Zero and Round to $+\infty$ can be written using only integer operations. The generated reciprocal is verified against this behavioural using Synopsys HECTOR technology, a formal equivalence checking tool. For logarithm and exponential, we verified that the hardware generated a result between our Python generated bounds using HECTOR.

Table \ref{tab:results_table} presents logic synthesis results for a number of fixed-point designs and compares them against the industrial state of the art, Synopsys DesignWare \cite{Synopsys2021DesignS-2021.06-SP2}. On average the proposed implementations improve the area-delay product by 7\%. In floating point implementations of functions such as reciprocal and logarithm, the piecewise polynomial approximation is the resource intensive computation since exponent handling is comparatively cheap. These designs could easily be combined with parameterised exponent handling code to generate complete floating point architectures.

Figure \ref{fig:single_prec_recip} presents complete area-delay profiles for the competing 23 bit implementations of the reciprocal function. We note that the proposed hardware is competitive across the delay spectrum offering area improvements at several delay targets. This is respectable since Designware is not static as the architecture selected by logic synthesis varies with delay.

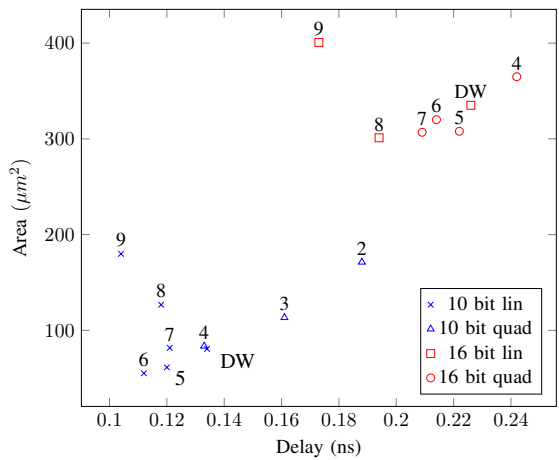
\begin{figure}
    \centering
    \begin{tikzpicture}[scale=0.75]
	\begin{axis}[%
	scatter/classes={%
		10={mark=x,blue},%
		11={mark=triangle,blue},%
		16={mark=square,red},%
		17={mark=o,red}},%
		xlabel=Delay (ns),ylabel=Area $(\mu m^2)$,%
		legend pos=south east]
	\addplot[scatter,only marks,%
		scatter src=explicit symbolic,%
		nodes near coords*={\Label},%
        visualization depends on={value \thisrow{label} \as \Label}]%
	table[meta=bw] {
        delay area bw label
        0.188 171.259921 11 2
        0.161 113.338801 11 3
        0.133 83.530080 11 4
        0.112  55.144080 10 6 
        0.121  81.779040 10 7 
         0.118 126.758881 10 8 
         0.104 179.905681 10 9 
         0.242 364.914002 17 4 
         0.222 307.909442 17 5 
         0.214 320.084642 17 6 
         0.209 306.856081 17 7 
         0.194 301.124162 16 8 
         0.173 400.700884 16 9 
         0.226 335.077922 16 DW
	};
	
	\addplot[scatter,only marks,%
		scatter src=explicit symbolic,%
		every node near coord/.style={anchor=270}*={\Label},%
        visualization depends on={value \thisrow{label} \as \Label}]%
	table[meta=bw] {
        area delay bw label
        0.134  80.670960 10 DW
        0.120  61.341120 10 5
	};
	
	\legend{10 bit lin, 10 bit quad, 16 bit lin, 16 bit quad}
	\end{axis}
	\draw [] (2.3,0.8) node [right] {\footnotesize DW};
	\draw [] (1.5,0.5) node [right] {\footnotesize 5};
\end{tikzpicture}
    \caption{Area-delay points at the minimum obtainable delay target for competing 10 and 16 bit implementations of the base 2 logarithm. Point labels indicate the number of lookup bits used (or DesignWare).}
    \label{fig:lut_plot}
\end{figure}

The exponential runtime scaling discussed in \S \ref{sec:dsg} is apparent in these results with the 23 bit reciprocal approaching the computational limit. In practice, this is not concerning as piecewise polynomial methods are rarely used for high precision approximations. 

One advantage is the ability to easily explore different LUT height architectures. Figure \ref{fig:lut_plot} highlights the challenge of optimising LUT height according to different metrics.

\begin{table}
    \centering
    \caption{Comparison against FloPoCo generated quadratic interpolation architectures with equal LUT height. Table dimensions [$a$ width, $b$ width, $c$ width] = total width.}
    \begin{tabular}{|c|c|c|c|}
         \hline
         Function & Bitwidth & FloPoCo LUT      & Proposed LUT \\
         \hline
         Recip    & 23       & [10,18,26] = 54  & [8,17,37] = 62 \\
         Log$_2$  & 16       & [ 8,15,20] = 43  & [4,14,20] = 38  \\
         Exp    & 10       & [ 6,11,14] = 31  & [2, 9,15] = 26    \\
         \hline
    \end{tabular}

    \label{tab:flopo_comparison}
\end{table}

To compare against the Remez algorithm, we generated equivalent implementations using FloPoCo \cite{deDinechin2011CustomGenerator}, with the results presented in Table \ref{tab:flopo_comparison}. FloPoCo generates narrower tables for the larger bitwidth at the expense of wider $a$ values than  produced here, resulting in larger $a\times x^2$ multiplication arrays. FloPoCo runs in seconds for all testcases. FloPoCo targets FPGAs, so it is not a relevant logic synthesis comparison point. 

\section{Conclusion}
This paper demonstrates a method to generate the complete design space of piecewise polynomial approximations to a complex function for arbitrary accuracy specifications using the given architecture. We generated RTL approximations to the reciprocal, base two logarithm and base two exponential functions and showed that they were competitive with state of the art implementations. Knowledge of the complete design space facilitates easy re-targeting. Generating the complete design space is computationally expensive and therefore only suitable up to binary32 implementations.

Future work will investigate a decision procedure to choose the optimal number of lookup bits. Integration with MPFR would provide arbitrary precision and trusted bounds. Scalability concerns could be addressed by introducing parallelism. 

\bibliographystyle{IEEEtran}
\bibliography{references.bib}

\begin{thebibliography}{10}
\providecommand{\url}[1]{#1}
\csname url@samestyle\endcsname
\providecommand{\newblock}{\relax}
\providecommand{\bibinfo}[2]{#2}
\providecommand{\BIBentrySTDinterwordspacing}{\spaceskip=0pt\relax}
\providecommand{\BIBentryALTinterwordstretchfactor}{4}
\providecommand{\BIBentryALTinterwordspacing}{\spaceskip=\fontdimen2\font plus
\BIBentryALTinterwordstretchfactor\fontdimen3\font minus
  \fontdimen4\font\relax}
\providecommand{\BIBforeignlanguage}[2]{{%
\expandafter\ifx\csname l@#1\endcsname\relax
\typeout{** WARNING: IEEEtran.bst: No hyphenation pattern has been}%
\typeout{** loaded for the language `#1'. Using the pattern for}%
\typeout{** the default language instead.}%
\else
\language=\csname l@#1\endcsname
\fi
#2}}
\providecommand{\BIBdecl}{\relax}
\BIBdecl

\bibitem{Muller1994Bkm:AFunctions}
J.~M. Muller, ``{Bkm:A New Hardware Algorithm for Complex Elementary
  Functions},'' \emph{IEEE Transactions on Computers}, vol.~43, no.~8, 1994.

\bibitem{Volder1959TheTechnique}
J.~E. Volder, ``{The CORDIC Trigonometric Computing Technique},'' \emph{IRE
  Transactions on Electronic Computers}, vol. EC-8, no.~3, 1959.

\bibitem{Tang1991Table-lookupAnalysis}
P.~T.~P. Tang, ``{Table-lookup algorithms for elementary functions and their
  error analysis},'' in \emph{Proceedings - Symposium on Computer Arithmetic},
  1991.

\bibitem{Drane2012CorrectlyMultiply-add}
T.~Drane, W.~C. Cheung, and G.~Constantinides, ``{Correctly rounded constant
  integer division via multiply-add},'' in \emph{IEEE International Symposium
  on Circuits and Systems}, 2012.

\bibitem{Lee2005OptimizingEvaluation}
D.~U. Lee, A.~A. Gaffar, O.~Mencer, and W.~Luk, ``{Optimizing hardware function
  evaluation},'' \emph{IEEE Transactions on Computers}, vol.~54, no.~12, 2005.

\bibitem{Detrey2005Table-basedEvaluation}
J.~Detrey and F.~De~Dinechin, ``{Table-based polynomials for fast hardware
  function evaluation},'' in \emph{Proceedings - International Conference on
  Application-Specific Systems, Architectures and Processors}, 2005.

\bibitem{Strollo2011ElementaryApproximations}
A.~G.~M. Strollo, D.~De~Caro, and N.~Petra, ``{Elementary functions hardware
  implementation using constrained piecewise-polynomial approximations},''
  \emph{IEEE Transactions on Computers}, vol.~60, no.~3, 2011.

\bibitem{deDinechin2011CustomGenerator}
F.~de~Dinechin and B.~Pasca, ``{Custom Arithmetic Datapath Design for FPGAs
  using the FloPoCo Core Generator},'' \emph{Design {\&} Test of Computers,
  IEEE}, vol.~PP, no.~99, 2011.

\bibitem{DeDinechin2010AutomaticEvaluation}
F.~De~Dinechin, M.~Joldes, and B.~Pasca, ``{Automatic generation of
  polynomial-based hardware architectures for function evaluation},'' in
  \emph{Proceedings - International Conference on Application-Specific Systems,
  Architectures and Processors}, 2010.

\bibitem{Chevillard2010Sollya:Codes}
S.~Chevillard, M.~Jolde{\c{s}}, and C.~Lauter, ``{Sollya: An environment for
  the development of numerical codes},'' in \emph{Lecture Notes in Computer
  Science}, vol. 6327 LNCS, 2010.

\bibitem{Brisebarre2007EfficientL-approximations}
N.~Brisebarre and S.~Chevillard, ``{Efficient polynomial L-approximations},''
  in \emph{Proceedings - Symposium on Computer Arithmetic}, 2007.

\bibitem{Synopsys2021DesignS-2021.06-SP2}
{Synopsys}, ``{Design Compiler User Guide S-2021.06-SP2},'' Synopsys, Mountain
  View, Tech. Rep., 6 2021.

\end{thebibliography}
\end{document}